\documentclass[submission,copyright,creativecommons]{eptcs}
\providecommand{\event}{Gandalf 2012} 
\usepackage{breakurl}             
\usepackage{mls}
\usepackage{amssymb}
\usepackage{amsfonts}
\usepackage{amsmath}
\usepackage{amsthm}
\usepackage{enumitem}
\usepackage{color}
\usepackage{setspace}
\usepackage{color}

\newtheorem{theorem}{Theorem}
\newtheorem{lemma}{Lemma}
\newtheorem{corollary}{Corollary}
\newtheorem{definition}{Definition}

\newcommand{\Lang}{\ensuremath{\mathbf{\forall}^{\pi}_{0,2}}\xspace}
\newcommand{\LangBounded}[1]{\ensuremath{(\Lang)^{\leq #1}}\xspace}

\newcommand{\assignment}[1]{M^{#1}}
\newcommand{\pairf}[1]{\pi^{#1}}

\newcommand{\inter}{I}

\newcommand{\ipairf}{\pairf{\inter}}
\newcommand{\iassignment}{\assignment{\inter}}

\newcommand{\interp}{J}

\newcommand{\ipairfp}{\pairf{\interp}}
\newcommand{\iassignmentp}{\assignment{\interp}}

\newcommand{\myTimes}{\mathrel{\times_{\pi}}}
\newcommand{\pairs}[2]{\mathrm{Pairs}_{\pi^{#1}}(#2)}

\newcommand{\Langdom}{\ensuremath{\mathbf{\forall}^{\pi +2 \dom}_{0,2}}\xspace}
\newcommand{\myDelta}{D_{\varphi}} 
\newcommand{\D}{\corr{\mathsf{U}}}

\newcommand{\nat}{\mathbb{N}}
\newcommand{\dominoSys}{\mathbb{D}}

\newcommand{\peanSys}{\mathbb{S}}
\newcommand{\peanN}{\mathcal{N}}
\newcommand{\peanZ}{\mathcal{Z}}
\newcommand{\peanS}{\mathcal{S}}
\newcommand{\isPeanoSys}{\mathtt{is\_Peano}}

\newcommand{\Partition}{P}
\newcommand{\isPartition}{\mathsf{partition}}
\newcommand{\partn}{A}
\newcommand{\mN}{\mathsf{N}}
\newcommand{\mQ}{\mathsf{Q}}
\newcommand{\mS}{\mathsf{S}}
\newcommand{\mZ}{\mathsf{Z}}
\newcommand{\horcom}{\mathsf{hor}}
\newcommand{\vercom}{\mathsf{ver}}

\newcommand{\vx}{x}
\newcommand{\vy}{y}

\newcommand{\svx}{x}
\newcommand{\svy}{y}
\newcommand{\svz}{z}
\newcommand{\mvx}{f}
\newcommand{\mvy}{g}
\newcommand{\mvz}{h}

\newcommand{\sx}{u}
\newcommand{\sy}{v}

\newcommand{\corr}[1]{#1}

\title{A decidable quantified fragment of set theory with ordered
pairs and some undecidable extensions\thanks{Work partially supported
by the INdAM/GNCS 2012 project
\emph{``Specifiche insiemistiche ese\-guibili e loro verifica 
formale''} and by Network Consulting Engineering Srl.}} 

\author{Domenico Cantone
\institute{Department of Mathematics and Computer Science\\
University of Catania, Italy}
\email{cantone@dmi.unict.it}
\and
Cristiano Longo
\institute{Network Consulting Engineering\\
Valverde, Catania, Italy} \email{cristiano.longo@nce.eu} }
\def\titlerunning{A decidable quantified fragment of set theory with
ordered pairs and some undecidable extensions} 
\def\authorrunning{D. Cantone \& C. Longo}
\begin{document}


\maketitle

\begin{abstract}
In this paper we address the decision problem for a fragment of set
theory with restricted quantification which extends the language
studied in \cite{BreFerOmoSch1981} with pair related quantifiers and
constructs, in view of possible applications in the field of
\emph{knowledge representation}.
We will also show that the decision problem for our language has a
non-deterministic exponential time complexity.  However, for the
restricted case of formulae whose quantifier prefixes have length
bounded by a constant, the decision problem becomes
\textsc{NP}-complete.  We also observe that in spite of such
restriction, several useful set-theoretic constructs, mostly related
to maps, are expressible.  Finally, we present some undecidable
extensions of our language, involving any of the operators domain,
range, image, and map composition.
\end{abstract}

\section{Introduction}\label{INTRO}

The intuitive formalism of set theory has helped providing solid and
unifying foundations to such diverse areas of mathematics as geometry,
arithmetic, analysis, and so on.
Hence, positive solutions to the decision problem for fragments of
set theory can have considerable applications to the automation of
mathematical reasoning and therefore in any area which can take
advantage of automated deduction capabilities.

The decision problem in set theory has been intensively studied
in the context of \emph{Computable Set Theory} (see
\cite{CanFerOmo89a,CanOmoPol01,SchCanOmo11}), and decision
procedures or undecidability results have been provided for several
sublanguages of set theory.
\emph{Multi-Level Syllogistic} (in short \mls, cf.\
\cite{FerOmoSch1980}) was the first unquantified sublanguage of set
theory that has been shown to have a solvable satisfiability problem.
We recall that \mls is the Boolean combinations of atomic formulae
involving the set predicates $\in$, $\subseteq$, $=$, and the Boolean
set operators $\cup$, $\cap$, $\setminus$.  Numerous
extensions of \mls with various combinations of operators (such as
singleton, powerset, unionset, etc.)  and predicates (on finiteness,
transitivity, etc.)  have been proved to be decidable.  Sublanguages
of set theory admitting explicit quantification (see for example
\cite{BreFerOmoSch1981, OmoPol2010, OmoPol2012, CanLonNic2011}) are of
particular interest, since, as reported in \cite{BreFerOmoSch1981},
they allow one to express several set-theoretical constructs using
only the basic predicates of membership and equality among sets.

Applications of Computable Set Theory to \emph{knowledge
representation} have been recently investigated in
\cite{CanLonPis2010, CanLonNic2011}, where some interrelationships
between (decidable) fragments of set theory and description logics
have been exploited.\footnote{We recall that description logics are a
well-established framework for knowledge representation; see
\cite{DLHANDBOOK2} for an introduction.}
As knowledge representation mainly focuses on representing
relationships among items of a particular domain, any set-theoretical
language of interest to knowledge representation should include a
suitable collection of operators on \emph{multi-valued
maps}.
\footnote{According to \cite{SchDewSchDub1986}, we use the term
`maps' to denote sets of ordered pairs.}

Non-deterministic exponential time decision procedures for two
unquantified fragments of set theory involving map related constructs
have been provided in \cite{FOS80, CanSch91}.  As in both
cases the map domain operator is allowed together with all the
constructs of \mls, it turns out that both fragments have an
\textsc{ExpTime}-hard decision problem (cf.\ \cite{CanLonNic2010}).
On the other hand, the somewhat less expressive fragment \mlsscart has
been shown to have an \textsc{NP}-complete decision problem in
\cite{CanLonNic2010}, where \mlsscart is a two-sorted
language with set and map variables, which involves various map
constructs like Cartesian product, map restrictions, map inverse, and
Boolean operators among maps, and predicates for single-valuedness,
injectivity, and bijectivity of maps.

In \cite{BreFerOmoSch1981}, an extension of the quantified fragment
$\forall_{0}$ (studied in the same paper---here the subscript `$0$' 
denotes that quantification is restricted) with \emph{single-valued}
maps, the map domain operator, and terms of the form $f(t)$, with $t$
a function-free term, was considered.  We recall that
$\forall_{0}$-formulae are propositional combinations of restricted
quantified prenex formulae
$(\forall y_1 \in z_1) \cdots (\forall y_n \in z_n)p$\,,
where $p$ is a Boolean combination of atoms of the types $x \in y$,
$x=y$, \corr{and \emph{quantified variables nesting} is not allowed, in the 
sense that any quantified variable $y_i$ can not occur at the
right-hand side of a membership symbol $\in$ in the same quantifier prefix
(roughly speaking, no $z_j$ can be a $y_i$)}. 
%
More recently, a decision procedure for a new fragment of set theory,
called $\Forallpizero$, has been presented in \cite{CanLonNic2011}.
\corr{The superscript ``$\pi$'' denotes the presence of operators
related to ordered pairs.} Formulae of the fragment $\Forallpizero$,
to be reviewed in Section \ref{DECPROC}, involve the operator
$\nonpairs{\cdot}$, which intuitively represents the collection of the
non-pair members of its argument, and terms of the form $[x,y]$, for
ordered pairs.  The predicates $=$ and $\in$ allowed in it can occur
only within atoms of the forms $x=y$, $x \in \nonpairs{y}$, and $[x,y]
\in z$; quantifiers in \Forallpizero-formulae are restricted to the
forms $(\forall x \in \nonpairs{y})$ and $(\forall [x,y] \in z)$, and,
much as in the case of the fragment $\forall_{0}$, quantified
variables nesting is not allowed.

In this paper we solve the decision problem for the
extension \Lang of the fragment $\forall_{0}$ with ordered pairs
and prove that, under particular conditions, our decision procedure
runs in non-deterministic polynomial time.
\Lang is a two-sorted \corr{(as indicated by the second subscript ``$2$'')} quantified fragment of set theory which allows
restricted quantifiers of the forms $(\forall \svx \in \svy)$,
$(\exists \svx \in \svy)$, $(\forall [\svx, \svy] \in \mvx)$,
$(\exists [\svx, \svy] \in \mvx)$, and literals of the forms $\svx \in
\svy$, $[\svx, \svy] \in \mvx$, $\vx=\vy$, $\mvx = \mvy$, where
$\svx$, $\svy$ are set variables and $\mvx$, $\mvy$ are map variables.
Considerably many set-theoretic constructs are expressible in it, as
shown in Table \ref{SETCONS}.  In fact, the language \Lang is also an
extension of \mlsscart.  However, as will be shown in Section
\ref{UNDEC}, it is not strong enough to express inclusions like $\svx
\subseteq \dom(\mvx)$, $\svx \subseteq \range(\mvx)$, $\svx \subseteq
\mvx[\svy]$, and $h \subseteq f \circ g$, but only those in which the
operators domain, range, \corr{(multi-)}image, and map composition are allowed to
appear on the left-hand side of the inclusion operator $\subseteq$.


The paper is organized as follows.  Section \ref{PREL} provides some
preliminary notions and definitions.  In Section \ref{LANG} we give
the precise syntax and semantics of the language \Lang.  Decidability
and complexity of reasoning in the language \Lang are addressed
in Section \ref{DECPROC}.  Some undecidable extensions of \Lang are
then presented in Section \ref{UNDEC}.  Finally, in Section \ref{CONC}
we draw our conclusions and provide some hints for future works.

\begin{table}[h]
\begin{center}
\begin{small}
\begin{tabular}{|c|l|}
\hline
&\\[-.3cm]
$\svx =\emptyset$ & $(\forall \svx' \in \svx)(x' \neq x')$ \\[0.0cm]

$\svx \subseteq \svy$ & $(\forall \svx' \in \svx)(\svx' \in \svy)$
\\[0.0cm]

$\svx=\svy \cup \svz$ & $\svy \subseteq \svx \wedge \svz \subseteq \svx \wedge (\forall \svx' \in \svx)(\svx' \in \svy \vee
\svx' \in \svz)$\\[0.0cm]

$\svx=\svy \cap \svz$ & $\svx \subseteq \svy \wedge \svx \subseteq \svz \wedge (\forall \svy' \in \svy)
(\svy' \in \svz \rightarrow \svy' \in \svx)$\\[0.0cm]

$\svx=\svy \setminus \svz$ & $\svx \subseteq \svy \wedge (\forall \svy' \in \svy)
(\svy' \in \svx \leftrightarrow \svy' \notin \svz)$\\[0.0cm]

$\svx=\{\svy\}$ & $\svy \in \svx \wedge (\forall \svx' \in \svx)(\svx'=\svy)$\\[0.0cm]

$\mvx = \emptyset$ & $(\forall [\svx, \svy] \in \mvx)(x \neq x)$ \\[0.0cm]

$\mvx \subseteq \mvy$ & $(\forall [\svx, \svy] \in \mvx)([\svx, \svy] \in \mvy)$ \\[0.0cm]

$\mvx=\mvy \cup \mvz$ & $\mvy \subseteq \mvx \wedge \mvz \subseteq \mvx \wedge (\forall [\svx,\svy] \in \mvx)([\svx,\svy] \in \mvy \vee [\svx,\svy] \in \mvz)$\\[0.0cm]

$\mvx=\mvy \cap \mvz$ & $\mvx \subseteq \mvy \wedge \mvx \subseteq \mvz \wedge (\forall [\svx,\svy] \in
\mvy)([\svx,\svy] \in \mvz \rightarrow [\svx, \svy] \in \mvx)$\\[0.0cm]

$\mvx=\mvy \setminus \mvz$ & $\mvx \subseteq \mvy \wedge
(\forall [\svx,\svy] \in \mvy)([\svx,\svy] \in \mvx \leftrightarrow [\svx,\svy] \notin \mvz)$\\[0.0cm]

$\mvx=\{[\svx, \svy]\}$ & $[\svx, \svy] \in \mvx \wedge (\forall [\svx', \svy'] \in \mvx)(\svx'=\svx \wedge \svy'=\svy)$\\[0.0cm]

$\mvx=\mvy^{-1}$ & $(\forall [\svx, \svy] \in \mvx)([\svy, \svx] \in \mvy) \wedge (\forall [\svx, \svy] \in \mvy)([\svy, \svx] \in \mvx)$\\[0.0cm]

$\mvx=\svx \times \svy$ & $(\forall \svx' \in \svx)(\forall \svy' \in \svy)([\svx',\svy'] \in \mvx) \wedge
  (\forall [\svx',\svy'] \in \mvx)(\svx' \in \svx \wedge \svy' \in \svy)$\\[0.0cm]

$\mvx=\mvy_{\svx|}$ & $\mvx \subseteq \mvy \wedge (\forall [\svx',\svy'] \in \mvy)([\svx',\svy'] \in \mvx \leftrightarrow \svx' \in \svx)$\\[0.0cm] 

$\mvx=\mvy_{|\svy}$ & $\mvx \subseteq \mvy \wedge (\forall [\svx',\svy'] \in \mvy)([\svx',\svy'] \in \mvx \leftrightarrow \svy' \in \svy)$\\[0.0cm] 

$\mvx=\mvy_{\svx|\svy}$ & $\mvx \subseteq \mvy \wedge (\forall [\svx',\svy'] \in \mvy)([\svx',\svy'] \in \mvx \leftrightarrow \svx' \in \svx \wedge \svy' \in \svy)$\\[0.0cm]

$\mvx=\identity{\svx}$ & $(\forall \svx' \in \svx)([\svx',\svx'] \in f) \wedge (\forall [\svx',\svy'] \in \mvx)(\svx'=\svy' \wedge \svx' \in \svx)$\\[0.0cm]

$\mvx=\sym(\mvy)$ & $(\forall [\svx, \svy] \in \mvx)([\svx, \svy] \in \mvy \vee [\svy, \svx] \in \mvy) \wedge 
  (\forall [\svx, \svy] \in \mvy)([\svx, \svy] \in \mvx \wedge [\svy, \svx] \in \mvx)$\\[0.0cm]

$\singlev(\mvx)$ & $(\forall [\svx,\svy] \in \mvx)(\forall [\svx',\svy'] \in \mvx)(\svx=\svx' \rightarrow \svy=\svy')$\\[0.0cm]

$\inj(\mvx)$ & $(\forall [\svx,\svy] \in \mvx)(\forall [\svx',\svy'] \in \mvx)(\svy=\svy' \rightarrow \svx=\svx')$\\[0.0cm]

$\bij(\mvx)$ & $(\forall [\svx,\svy] \in \mvx)(\forall [\svx',\svy'] \in \mvx)(\svx=\svx' \leftrightarrow \svy=\svy')$\\[0.0cm]

$\isTrans(\mvx)$ & $(\forall [\svx, \svy] \in \mvx)(\forall [\svx', \svy'] \in \mvx)(\svy=\svx' \rightarrow [\svx, \svy'] \in \mvx)$\\[0.0cm]


$\isIrrefl(\mvx)$ & $(\forall [\svx, \svy] \in \mvx)(\svx \neq \svy)$\\[0.0cm]

$\isAsym{\mvx}$ & $(\forall [\svx, \svy] \in \mvx)(\svx=\svy \vee [\svy, \svx] \notin \mvx)$\\[0.0cm]

$\mvx \circ \mvy \subseteq \mvz$ & $(\forall [\svx, \svy] \in \mvx)(\forall [\svx',\svy'] \in \mvy)(\svy=\svx' \rightarrow [\svx, \svy'] \in \mvz)$\\[0.0cm]

$\dom(\mvx) \subseteq \svx$ & $(\forall [\svx', \svy'] \in \mvx)(\svx' \in \svx)$\\[0.0cm]

$\range(\mvx) \subseteq \svy$ & $(\forall [\svx', \svy'] \in \mvx)(\svy' \in \svy)$\\[0.0cm]

$\mvx[\svx] \subseteq \svy$ & $(\forall [\svx', \svy'] \in 
\mvx)(\svx' \in \svx \rightarrow \svy' \in \svy)$\\[.1cm]
\hline
\end{tabular}
\end{small}
\end{center}
\caption{Set-theoretic constructs expressible in \Lang.}\label{SETCONS}
\end{table}

\section{Preliminaries}\label{PREL}

We briefly review basic notions from set theory and introduce
also some definitions which will be used throughout the paper.

Let $\SetVars \defAs \{\svx, \svy, \svz, \ldots\}$ and
$\MapVars \defAs \{\mvx, \mvy, \mvz, \ldots\}$ be two infinite disjoint
collections of \emph{set} and \emph{map variables},
respectively.  As we will see, map variables will be interpreted as
maps (i.e., sets of ordered pairs).  We put $\Vars \defAs \SetVars 
\cup \MapVars$.
For a formula $\varphi$, we write $\Vars(\varphi)$ for the collection
of variables occurring free (i.e., not bound by any quantifier) in
$\varphi$, and put $\SetVars(\varphi) \defAs \Vars(\varphi) \cap
\SetVars$ and $\MapVars(\varphi) \defAs \Vars(\varphi) \cap \MapVars$.

Semantics of most of the languages studied in the context of
Computable Set Theory are based on the \emph{von Neumann standard
cumulative hierarchy of sets} $\VNU$, which is the class containing
all the \emph{pure} sets (i.e., all sets whose members are recursively
based on the empty set $\emptyset$).  The von Neumann hierarchy $\VNU$
is defined as follows:
\[
\begin{array}{rcll}
  \VNU_0          & \defAs & \emptyset
\\
  \VNU_{\gamma+1} & \defAs & \Powerset(\VNU_\gamma) \,,
  & \textrm{for each ordinal $\gamma$}
\\
  \VNU_\lambda    & \defAs & \bigcup_{\mu < \lambda} \VNU_\mu
  \,,& \textrm{for each limit ordinal $\lambda$}
\\
  \VNU & \defAs & \bigcup_{\gamma \in \mathit{On}} \VNU_{\gamma} \,,
\end{array}
\]
where $\Powerset(\cdot)$ is the powerset operator and
$\mathit{On}$ denotes the class of all ordinals. 
The \emph{rank}
$\rank(u)$ of a set $u \in \VNU$ is defined as the least ordinal
$\gamma$ such that $u \in \VNU_{\gamma}$.
%
We will refer to mappings from 
$\Vars$ to $\VNU$ as \emph{assignments}.

Next we introduce some notions related to pairing functions and
ordered pairs.
Let $\pi(\cdot, \cdot)$ be a binary operation over the universe \VNU.
The \emph{Cartesian product} $\sx \myTimes \sy$ of two sets $\sx, \sy \in
\VNU$, relative to $\pi$, is defined as
$\sx \myTimes \sy \defAs \{ \pi(\sx',\sy') : \sx' \in \sx \wedge \sy' \in \sy\}$.
When it is clear from the context, for the sake of conciseness we
will omit to specify the binary operation $\pi$ and simply write 
`$\times$' in place of `$\myTimes$'.
A binary operation $\pi$ over sets in \VNU  is said to be a 
\emph{pairing function} if
\begin{itemize}
    \item[(i)] 
    $\pi(\sx,\sy)=\pi(\sx',\sy') \iff \sx=\sx' \wedge \sy=\sy'\,, \quad \text{for
    all~} u,u',v,v' \in \VNU$, and

    \item[(ii)] the Cartesian product $\sx \times \sy$ (relative to
    $\pi$) is a set of \VNU, for all $\sx,\sy \in \VNU$.
\end{itemize}
In view of the replacement axiom, condition (ii) is obvioulsy met 
when $\pi(\sx,\sy)$ is expressible by a set-theoretic term. This, for 
instance, is the case for Kuratowski's ordered pairs, defined by 
$\pi_{\mathsf{Kur}}(\sx, \sy) \defAs \{ \{\sx\}, \{\sx, \sy\}\}$,
for all $\sx, \sy \in \VNU$.
Given a pairing function $\pi$ and a set $s$, we denote with
$\pairs{}{s}$ the collection of the \emph{pairs} in $s$ (with respect
to $\pi$), namely
$ \pairs{}{s} \defAs \{ u \in s : (\exists v_1, v_2 \in \VNU)(u =
 \pi(v_1,v_2)) \}$.

A \emph{pair-aware interpretation} $\inter=(\iassignment, \ipairf)$
consists of a pairing function $\ipairf$ and an assignment
$\iassignment$ such that
$\pairs{\inter}{\iassignment \mvx}=\iassignment \mvx$ 
holds for every map variable $\mvx \in \MapVars$ (i.e., map variables
can only be assigned sets of ordered pairs, or the empty set).  For
conciseness, in the rest of the paper we will refer to \emph{pair-aware}
interpretations just as interpretations.
An interpretation $\inter=(\iassignment, \ipairf)$ associates sets to
variables and pair terms, respectively, as follows:
\begin{equation}
    \label{interpVarsPairs}
\begin{array}{rcl}
 \inter \vx & \defAs & \iassignment \vx,\\
 \inter [\vx,\vy] & \defAs& \ipairf(\inter \vx, \inter \vy),
\end{array}
\end{equation}
for all $\vx,\vy \in \Vars$. 
Let $W\subseteq\Vars$ be a finite collection of variables, and let
$M,M'$ be two assignments.  We say that $M'$ is a $W$\emph{-variant}
of $M$ if $M x = M' x$ for all $x \in \Vars \setminus W$.  For two
interpretations $\inter = (\iassignment, \ipairf)$ and $\interp =
(\iassignmentp, \ipairfp)$, we say that $\interp$ is a $W$-variant of
$\inter$ if $\iassignmentp$ is a $W$-variant of $\iassignment$ and
$\ipairfp=\ipairf$.

In the next section we introduce the precise syntax and semantics of 
the language \Lang.

\section{The language \Lang}\label{LANG}

The language \Lang consists of the denumerable infinity of variables
$\Vars = \SetVars \cup \MapVars$, the binary \emph{pairing} operator
$[\cdot,\cdot]$, the predicate symbols $\in, =$, the Boolean
connectives of propositional logic $\neg$, $\wedge$, $\vee$,
$\rightarrow$, $\leftrightarrow$, parentheses, and \emph{restricted}
quantifiers of the forms $(\forall x \in y)$, $(\forall [x,y] \in f)$,
$(\exists x \in y)$, and $(\exists [x,y] \in f)$.
\emph{Atomic} \Lang\emph{-formulae} are expressions of
the following four types
\begin{equation}
    \label{atomic}
\svx \in \svy, \quad \svx=\svy, \quad  [\svx,\svy] \in \mvx, \quad \mvx=\mvy \,,
\end{equation}
with $\svx,\svy \in \SetVars$ and $\mvx, \mvy \in \MapVars$.
\emph{Quantifier-free} \Lang\emph{-formulae} are propositional
combinations of atomic \Lang-formulae. 
\emph{Prenex} \Lang\emph{-formulae} are expressions of the 
following two forms
\begin{gather}
(\forall \svx_1 \in \svz_1) \ldots (\forall \svx_h \in \svz_h)(\forall
 [\svx_{h+1}, \svy_{h+1}] \in \mvx_{h+1})\ldots(\forall [\svx_n, \svy_n] \in
 \mvx_n)\delta\,,  \label{UNIVFORM} \\
(\exists \svx_1 \in \svz_1) \ldots (\exists \svx_h \in \svz_h)(\exists
 [\svx_{h+1}, \svy_{h+1}] \in \mvx_{h+1})\ldots(\exists [\svx_n, \svy_n] \in
 \mvx_n)\delta\,, \label{EXFORM}
\end{gather}
where $\svx_i, \svy_i, \svz_i \in \SetVars$, $\mvx_j \in \MapVars$,
and $\delta$ is a quantifier-free \Lang-formula.  We will refer to the
variables $\svz_1, \ldots, \svz_h$ as the \emph{domain variables} of
the formulae (\ref{UNIVFORM}) and (\ref{EXFORM}).  Notice that
quantifier-free \Lang-formulae can also be regarded as prenex
\Lang-formulae with an empty quantifier prefix.
A prenex \Lang-formula is said to be \emph{simple} if nesting among
quantified variables is not allowed, i.e., if no quantified variable
can occur also as a domain variable.
Finally, \Lang\emph{-formulae} are Boolean combinations of
simple-prenex \Lang-formulae.


Semantics of \Lang-formulae is given in terms of interpretations.  An
interpretation $\inter=(\iassignment, \ipairf)$ \emph{evaluates} a
\Lang-formula $\varphi$ into a truth value $\inter \varphi \in
\{\true,\false\}$ in the following recursive manner.
First of all, interpretation of quantifier-free \Lang-formulae is
carried out following the rules of propositional logic, where atomic
formulae (\ref{atomic}) are interpreted according to the standard
meaning of the predicates $\in$ and $=$ in set theory and the pair
operator $[\cdot,\cdot]$ is interpreted as in (\ref{interpVarsPairs}).
Thus, for instance, $\inter ([\svx,\svy] \in \mvx \rightarrow x \in y
) = \true$, provided that either $\ipairf(\inter \vx, \inter \vy)
\notin \inter \mvx$ or $\inter \vx \in \inter \vy$.
Then, evaluation of simple-prenex \Lang-formulae is defined
recursively as follows:
\begin{itemize}
  \item $\inter (\forall \svx \in \svz)\varphi = \true$, provided that
  $\interp \varphi = \true$, for every $\{\svx\}$-variant $\interp$ of
  $\inter$ such that $\interp \svx \in \interp \svz$;
  
  \item $\inter (\forall [\svx,\svy] \in \mvx)\varphi = \true$,
  provided that $\interp \varphi = \true$, for every
  $\{\svx,\svy\}$-variant $\interp$ of $\inter$ such that $\interp
  [\svx,\svy] \in \interp \mvx$;  
  
  \item $\inter (\exists \svx \in \svz)\varphi = \true$, provided that
  $\inter (\forall \svx \in \svz)\neg \varphi =\false$; and 
  
  \item $\inter (\exists [\svx,\svy] \in \mvx)\varphi = \true$,
  provided that $\inter (\forall [\svx,\svy] \in \mvx) \neg \varphi = 
  \false$.
\end{itemize}
Finally, evaluation of \Lang-formulae is carried out following the
rules of propositional logic.  

If an interpretation $\inter$ evaluates a \Lang-formula to $\true$ we
say that $\inter$ is a \emph{model} for $\varphi$ (and write $\inter
\models \varphi$).  
A \Lang-formula $\varphi$ is said to be \emph{satisfiable} if and only
if it admits a model.  
Two \Lang-formulae are said to be
\emph{equivalent} if they have exactly the same models.
Two \Lang-formulae $\varphi$ and $\varphi'$ are
said to be \emph{equisatisfiable} provided that $\varphi$ is
satisfiable if and only if so is $\varphi'$.  
The \emph{satisfiability
problem} (s.p., for short) for the theory \Lang is the problem of
establishing algorithmically whether any given \Lang-formula is
satisfiable or not.

By way of a simple normalization procedure based on disjunctive normal
form, the s.p.\ for \Lang-formulae can be reduced to
that for \emph{conjunctions} of simple-prenex \Lang-formulae of the
types (\ref{UNIVFORM}) and (\ref{EXFORM}).  Moreover, since any such
conjunction of the form
\[
\psi \wedge (\exists \svx_1 \in \svz_1) \ldots (\exists \svx_h \in
\svz_h)(\exists [\svx_{h+1}, \svy_{h+1}] \in \mvx_{h+1})\ldots(\exists
[\svx_n, \svy_n] \in \mvx_n)\delta
\]
is equisatisfiable with $\psi \wedge \delta_{+}' $, where
$\delta_{+}'$ is obtained from the quantifier-free formula
\[
\delta_{+} \defAs \bigwedge_{i=1}^{h} \svx_i \in \svz_i 
\wedge \bigwedge_{j=h+1}^{n} [\svx_{j}, \svy_{j}] \in \mvx_{j}
\wedge \delta
\]
by a suitable
renaming of the (quantified) variables
$x_{1},\ldots,x_{n},y_{h+1},\ldots,y_{n}$, 
it turns out that the s.p.\ for \Lang-formulae can be
reduced to the s.p.\ for \emph{conjunctions} of
simple-prenex \Lang-formulae of the type (\ref{UNIVFORM}) only, which
we call \emph{normalized \Lang-conjunctions}.

Satisfiability of normalized \Lang-conjunctions does not depend
strictly on the pairing function of the interpretation, provided that
suitable conditions hold, as proved in the following technical lemma.

\begin{lemma}\label{PFISO}
Let $\varphi$ be a normalized \Lang-conjunction, and let $\inter$ and
$\interp$ be two interpretations such that
\begin{enumerate}[label=(\alph*)]
 \item\label{PFISO_a} $\inter \vx = \interp \vx$, for all $\vx \in 
 \SetVars$,

 \item\label{PFISO_b} $\ipairf(u,v) \in \inter \mvx \iff \ipairfp(u,v)
 \in \interp \mvx$, for all $u,v \in \VNU$ and $\mvx \in \MapVars$.
\end{enumerate}
Then $\inter \models \varphi \iff \interp \models \varphi$.
\end{lemma}
\begin{proof}
It is enough to prove that
\begin{equation}\label{PFISO1}
 \inter \models \psi \iff \interp \models \psi
\end{equation}
holds, for every (universal) simple-prenex conjunct $\psi$ occurring
in $\varphi$.  We shall proceed by induction on the length of the
quantifier prefix of $\psi$.  We begin with observing that, by
\ref{PFISO_a}, $\inter$ and $\interp$ evaluate to the same truth
values all atomic formulae of the types $\svx \in \svy$ and
$\svx=\svy$, for all $\svx, \svy \in \SetVars$.  Likewise,
\[
\inter \models \mvx=\mvy \iff \interp \models \mvx=\mvy
\qquad \text{and} \qquad
\inter \models [\svx, \svy] \in \mvx \iff \interp \models
 [\svx, \svy] \in \mvx
\]
follow directly from \ref{PFISO_a} and \ref{PFISO_b}.  Thus
(\ref{PFISO1}) follows easily when $\psi$ is quantifier-free, i.e., 
when the length of its quantifier prefix is $0$.

Next, let $\psi=(\forall \svx \in \svy)\psi_0$, for some $\svx, \svy
\in \SetVars$, where $\psi_0$ is a universally quantified
simple-prenex \Lang-formula with one less quantifier than $\psi$ and
containing no quantified occurrence of $\svy$.  We prove that
$\inter_{u}$ is a model for $\psi_0$ if and only if so is
$\interp_{u}$, for every $u \in \inter \svy = \interp \svy$, where
$\inter_{u}$ and $\interp_{u}$ denote, respectively, the $\{ \svx
\}$-variants of
$\inter$ and $\interp$ such that $\inter_{u} x =
\interp_{u} x = u$.  But, for each $u \in \inter \svy = \interp \svy$,
$\inter_{u}$ and $\interp_{u}$ satisfy conditions \ref{PFISO_a} and
\ref{PFISO_b} of the lemma, so that, by inductive hypothesis, we have
$\inter_{u} \models \psi_0 \iff \interp_{u} \models \psi_0$.  Hence
$\inter \models (\forall \svx \in \svy)\psi_0 \iff \interp \models
(\forall \svx \in \svy)\psi_0$.

The case in which $\psi=(\forall [\svx, \svy] \in \mvx)\psi_0$, with
$\svx, \svy \in \SetVars$, $\mvx \in \MapVars$, and $\psi_0$ a
universally quantified simple-prenex \Lang-formula containing no
quantified occurrence of $\svx$ and $\svy$, can be dealt with much in
the same manner, thus concluding the proof of the lemma.
\end{proof}

In the following section we show that the s.p.\ for normalized
\Lang-conjunctions is solvable. 

\section{A decision procedure for \Lang}\label{DECPROC}

We solve the s.p.\ for \Lang-formulae by reducing the s.p.\ for
normalized \Lang-conjunctions to the s.p.\ for the fragment of set
theory \Forallpizero, studied in \cite{CanLonNic2011}.
Following \cite{CanLonNic2011},
\Forallpizero-formulae are finite conjunctions of \emph{simple-prenex
\Forallpizero-formulae}, namely expressions of the form
\[
 (\forall x_1 \in \nonpairs{z_1})\ldots(\forall x_h \in \nonpairs{z_h})(\forall [x_{h+1}, y_{h+1}] \in z_{h+1})\ldots(\forall [x_n, y_n] \in z_n)\delta,
\]
where $x_i, y_i, z_i \in \SetVars$, for $i = 1,\ldots,n$, no domain
variable $z_{i}$ can occur quantified, and $\delta$ is a
quantifier-free Boolean combination of atomic formulae of the types
$x \in \nonpairs{z}$, $[x,y] \in z$, $x=y$, 
with $x,y,z \in \SetVars$.\footnote{Thus, normalization is already
built-in into \Forallpizero-formulae, and we could have called them
\emph{normalized \Forallpizero-conjunctions}.}
Intuitively, a term of the form $\nonpairs{z}$ represents the set of 
the \emph{non-pair} members of $z$.
Notice that \Forallpizero-formulae involve only set variables.

Semantics for \Forallpizero-formulae is given by extending
interpretations also to terms of the form $\nonpairs{x}$ as indicated
below:
\[
  \inter \nonpairs{\svx} \defAs \inter \svx \setminus \pairs{\inter}{\inter
  \svx} \,,
\]
where $\svx \in \SetVars$.
Evaluation of \Forallpizero-formulae is carried out much in the same
way as for \Lang-formulae.  In particular, we also put
$\inter (\forall x \in \nonpairs{y})\varphi = \true$, provided
  that $\interp \varphi = \true$, for every $\{x\}$-variant $\interp$
  of $\inter$ such that $\interp x \in \inter \nonpairs{y}$. 

We recall that satisfiability of \Forallpizero-formulae can be tested
in non-deterministic exponential time.  Additionally, the s.p.\ for
\Forallpizero-formulae with quantifier prefixes of length at most $h$,
for any fixed constant $h \geq 0$, is \textsc{NP}-complete (cf.\
\cite{CanLonNic2011}).

The s.p.\ for normalized \Lang-conjunctions can be reduced to the
s.p.\ for \Forallpizero-formulae.  To begin with, we define a
syntactic transformation $\tau(\cdot)$ on normalized
\Lang-conjunctions.  More specifically, $\tau(\varphi)$ is obtained
from a given normalized \Lang-conjunction $\varphi$ by replacing
\begin{itemize}
    \item each restricted universal quantifier $(\forall x \in y)$ in
    $\varphi$ \corr{by} the quantifier $(\forall x \in \nonpairs{y})$,
    
    \item each atomic formula $x \in y$ in $\varphi$ \corr{by} the literal
    $x \in \nonpairs{y}$, and
    
    \item each map variable $\mvx$ occurring in $\varphi$ by a fresh
    set variable $\svx_{\mvx}$, thus identifying an application $\mvx
    \mapsto \svx_\mvx$ from $\MapVars(\varphi)$ into $\SetVars$, 
    which will be referred to as \emph{\corr{map-variable renaming} for 
    $\tau(\varphi)$}.
\end{itemize}    
Thus, for instance, if
\begin{align*}
\varphi &= (\forall \svx' \in \svx)([\svx,\svx] \in \mvx)
\wedge (\forall [\svx',\svy'] \in \mvx)(\svx'=\svy' \wedge \svx' \in
\svx)
\end{align*}
then
\[
\tau(\varphi) = (\forall \svx' \in \nonpairs{\svx})([\svx,\svx] \in
\svx_{\mvx}) \wedge (\forall [\svx',\svy'] \in
\svx_{\mvx})(\svx'=\svy' \wedge \svx' \in \nonpairs{\svx})\,,
\]
where $\svx_{\mvx}$ is a set variable distinct from $\svx$, $\svx'$, 
and $\svy'$.

The following lemma
provides a useful semantic relation between universal
simple-prenex \Lang-formulae and their corresponding
\Forallpizero-formula via $\tau$.


\begin{lemma}\label{SAT0}
Let $\psi$ be a universal simple-prenex \Lang-formula and let 
$\inter = (\iassignment, \ipairf)$
be an interpretation such that
\begin{enumerate}[label=(\roman{*}), ref=(\roman{*})]
  \item\label{SAT01} $\pairs{\inter}{\{ \inter \svx : \svx \in
  \SetVars(\psi)\}} = \emptyset$ (i.e., $\inter x$ is not a pair, for
  any free variable $x$ of $\psi$), and

  \item\label{SAT02} $\pairs{\inter}{\inter \svx} = \emptyset$, for
  every domain variable $\svx$ of $\psi$.
\end{enumerate}
Then $\inter \models \psi$ 
if and only if $\inter \models \tau(\psi)$.
\end{lemma}
\begin{proof}
We proceed by induction on the quantifier prefix length $\ell \geq 0$
of the formula $\psi$.  To begin with, we observe that in force of
\ref{SAT01} we have $\inter \svx \in \inter \svy$ if and only if
$\inter \svx \in \inter \nonpairs{\svy}$, for any two free variables
$x$ and $y$ of $\psi$, so that, given any atomic formula $\alpha$
involving only variables in $\SetVars(\psi)$, $\inter \models \alpha$
if and only if $\inter \models \tau(\alpha)$.  Hence the lemma follows
directly from propositional logic if $\psi$ is quantifier-free, i.e.,
$\ell = 0$.

Next, let $\psi=(\forall \svx \in \svy)\psi_0$, where $\psi_0$ is a
universal simple-prenex \Lang-formula with $\ell - 1$ quantifiers,
$\svx, \svy$ are set variables occurring neither as domain nor as
bound variables in $\psi_0$.  Observe that, by
\ref{SAT02}, $\inter \svy = \inter \nonpairs{\svy}$, since $\svy$ is a
domain variable of $\psi$.  Thus
it will be enough to prove that 
\begin{equation}\label{SAT03}
\inter_{\sx} \models \psi_0 \iff \inter_{\sx} \models \tau(\psi_0) 
\end{equation}
holds for every $\{\svx\}$-variant $\inter_{\sx}$ of $\inter$ such
that $\inter_{\sx} \svx=\sx$, with $\sx \in \inter \svy$.  But
$\inter_{\sx} \svx$ can not be a pair (with respect to the pairing
function $\ipairf$), as it is a member of $\inter \svy$ and $\svy$ is
a domain variable of $\varphi$.  Thus (\ref{SAT03}) follows by
applying the inductive hypothesis to $\psi_0$ and to every
interpretation $\inter_{\sx}$ such that $\sx \in \inter \svy$.

Finally, the case in which $\psi=(\forall [\svx, \svy] \in
\mvx)\psi_0$, where $\psi_0$ is a universal simple-prenex
\Lang-formula, $\svx, \svy$ are set variables not occurring as domain
variables in $\psi_0$, and $\mvx$ is a map variable, can be dealt with
much in the same way as the previous case, and is left to the
reader.
\end{proof}

In the following theorem we use the transformation $\tau(\cdot)$ to
reduce the s.p.\ for normalized \Lang-conjunctions to the s.p.\ for
\Forallpizero-formulae.

\begin{theorem}\label{DEC}
The s.p.\ for normalized \Lang-conjunctions can be reduced in linear
time to the s.p.\ for \Forallpizero-formulae, and therefore it is in
\textsc{NExpTime}.
\end{theorem}
\begin{proof}
We prove the theorem by showing that, given any normalized
\Lang-conjunction $\psi$, we can construct in linear time a
corresponding \Forallpizero-formula $\psi'$ which is
equisatisfiable with $\psi$.

So, let $\psi$ be a normalized \Lang-conjunction and let $\mvx
\mapsto \svx_\mvx$ be the map-variable renaming for $\tau(\psi)$.
We define the corresponding \Forallpizero-formula $\psi'$ as
follows:
\[
\psi' \defAs \tau(\psi) 
\wedge \bigwedge_{\svz \in \SetVars(\psi)}(\forall [\svx,\svy] \in
\svz)(\svx \neq \svx) 
\wedge \bigwedge_{\mvx \in \MapVars(\psi)} (\forall \svx \in
\nonpairs{\svx_\mvx})(\svx \neq \svx) 
\wedge \bigwedge_{\svz \in \SetVars(\psi)} (\svz \in \nonpairs{\D})
\,,
\]
where $\D$ is a fresh set variable. Plainly, the size of $\psi'$ 
is linear in the size of $\psi$.

Let us first assume that $\psi'$ admits a model $\interp =
(\iassignmentp, \ipairfp)$.  For each $\svz \in \SetVars(\psi)$ we
have $\pairs{\interp}{\interp \svz}=\emptyset$, as $\interp (\forall
[\svx,\svy] \in \svz)(\svx \neq \svx) = \true$, for $\svz \in
\SetVars(\psi)$.  Likewise, for each $\mvx \in \MapVars(\psi)$
we have $\interp \svx_\mvx = \pairs{\interp}{\interp \svx_\mvx}$, as
$\interp (\forall \svx \in \nonpairs{\svx_\mvx})(\svx \neq \svx) =
\true$, for $\mvx \in \MapVars(\psi)$.  Finally, for each $\svx \in
\SetVars(\psi)$, we have $\interp \svx \in \interp \D \setminus
\pairs{\interp}{\interp \D}$, so that $\pairs{\interp}{\{ \interp \svx
: \svx \in \SetVars(\psi)\}} = \emptyset$.
We define $\inter$ as the $\MapVars(\psi)$-variant of $\interp$
such that $\inter \mvx = \interp \svx_\mvx$, for $\mvx \in
\MapVars(\psi)$.  Plainly, $\inter \models \tau(\psi)$ so
that, by Lemma \ref{SAT0}, $\inter \models \psi$ as well.

For the converse direction, let
$\inter = (\iassignment, \ipairf)$ be a model for $\psi$.  We shall
exhibit an interpretation $\interp'$ which satisfies $\psi'$.  To
begin with, we define a new pairing function $\ipairfp$ by putting
\[
 \ipairfp(\sx, \sy) \defAs \{ \pi_{\mathsf{Kur}}(\sx, \sy), \{\myDelta\}
 \}\,,
\]
for every $\sx, \sy \in \VNU$, where $\pi_{\mathsf{Kur}}$ is the
Kuratowski's pairing function and $\myDelta \defAs \{\inter \svx :
\svx \in \SetVars(\psi)\}$.  Then we define $\iassignmentp$ as the
$\MapVars(\psi)$-variant of the assignment $\iassignment$ such that
$\iassignmentp \mvx = \{ \ipairfp(\sx, \sy) : \sx, \sy \in \VNU
\text{ and } \ipairf(\sx, \sy) \in \iassignment \mvx\}$, for each
$\mvx \in \MapVars(\psi)$.  From Lemma~\ref{PFISO}, it follows that
the interpretation $\interp = (\iassignmentp, \ipairfp)$ satisfies
$\psi$.  
Moreover, we have
\begin{equation}
    \label{PairsEq}
    \pairs{\interp}{\interp \svz}=\emptyset\,,
\end{equation}
for each $\svz \in \SetVars(\psi)$.  
Indeed, if for some $\sx, \sy \in \VNU$
and $\svz \in \SetVars(\psi)$ we had $\ipairfp(\sx, \sy) \in
\interp \svz$, then
\[
\inter z \in \myDelta \in \{\myDelta\} \in \{ \pi_{\mathsf{Kur}}(\sx, \sy), \{\myDelta\}
 \} = \ipairfp(\sx, \sy) \in \interp \svz = \inter z\,,
\]
contradicting the regularity axiom of set theory. 
Next, let $W \defAs \{\svx_{\mvx} : \mvx \in \MapVars(\psi)\} \cup
\{\D\}$ and let $\interp'$ be the $W$-variant of $\interp$, where
$\interp' \svx_{\mvx} = \interp \mvx$, for $\mvx \in
\MapVars(\psi)$, and $\interp' \D = \{\interp \svz : \svz \in
\SetVars(\psi)\}$ .  In view of (\ref{PairsEq}), it is an easy
matter to verify that 
\begin{equation}\label{secondEq}
\interp' \models \tau(\psi)\,.
\end{equation}
From (\ref{PairsEq}), we have immediately that
$\pairs{\interp'}{\interp' \svz}=\emptyset$, so that
\begin{equation}\label{firstEq}
\interp' \models \bigwedge_{\svz \in \SetVars(\psi)} (\forall
[\svx,\svy] \in \svz)(\svx \neq \svx)\,.
\end{equation}
Likewise, by reasoning much in
the same manner as for the proof of (\ref{PairsEq}), one can prove
that
\begin{equation}
    \label{thirdEq}
\interp' \models \bigwedge_{\mvx \in \MapVars(\psi)} (\forall \svx
\in \nonpairs{\svx_\mvx})(\svx \neq \svx) \wedge 
\bigwedge_{\svz \in \SetVars(\psi)} (\svz \in \nonpairs{\D})\,.
\end{equation}
From (\ref{secondEq}), (\ref{firstEq}), and (\ref{thirdEq}), it 
follows at once that $\interp' \models \psi'$, completing the 
proof that $\psi$ and $\psi'$ are equisatisfiable.

Since the s.p.\ for \Forallpizero-formulae is in \textsc{NExpTime}, 
as was shown in \cite[Section~3.1]{CanLonNic2011}, it readily follows 
that the s.p.\ for normalized \Lang-conjunctions is in 
\textsc{NExpTime} as well.
\end{proof}

\begin{corollary}
    \label{corollaryLangFormulae}
    The s.p.\ for \Lang-formulae is in \textsc{NExpTime}.
\end{corollary}
\begin{proof}
Let $\varphi$ be a satisfiable \Lang-formula.  We may assume without
loss of generality that all existential simple-prenex \Lang-formulae
of the form (\ref{EXFORM}) have already been rewritten in terms of
equivalent universal simple-prenex \Lang-formulae of the form
(\ref{UNIVFORM}), so that $\varphi$ is a propositional combination of
universal simple-prenex \Lang-formulae.  In addition, by suitably
renaming variables, we may assume that all quantified variables in
$\varphi$ are pairwise distinct and that they are also distinct from
free variables.

Let $\Sigma_{\varphi}=\{\psi_1, \ldots, \psi_n\}$ be the collection of
the universal simple-prenex \Lang-formulae occurring in $\varphi$.  By
traversing the syntax tree of $\varphi$, one can find in linear time
the propositional skeleton $P_{\varphi}$ of $\varphi$ and a
substitution $\sigma$ from the propositional variables
$\mathsf{p}_{1}, \ldots, \mathsf{p}_{n}$ of $P_{\varphi}$ into
$\Sigma_{\varphi}$, such that $P_{\varphi} \sigma = \varphi$, where
$P_{\varphi} \sigma$ is the result of substituting each propositional
variable $\mathsf{p}_{i}$ in $P_{\varphi}$ by the universal
simple-prenex \Lang-formula $\sigma(\mathsf{p}_{i})$.  Then to check
the satisfiability of $\varphi$ one can perform the following
non-deterministic procedure:
\begin{itemize}
    \item guess a Boolean valuation $\nu$ of the propositional 
    variables $\mathsf{p}_{1}, \ldots, \mathsf{p}_{n}$ of 
    $P_{\varphi}$ such that $\nu(P_{\varphi}) = \true$;
    
    \item form the \Lang-conjunction
    \begin{equation}\label{FORALLPIZERONU}
    \bigwedge_{\nu(\mathsf{p}_{i}) = \true}
    \sigma(\mathsf{p}_{i}) 
    \,\,\, \wedge \,\,\,
    \bigwedge_{\nu(\mathsf{p}_{i}) = \false}
    \neg\sigma(\mathsf{p}_{i}) \,;
    \end{equation}
    
    \item transform each conjunct 
    \[
    \neg (\forall \svx_1 \in \svz_1) \ldots (\forall \svx_h \in
    \svz_h)(\forall [\svx_{h+1}, \svy_{h+1}] \in
    \mvx_{h+1})\ldots(\forall [\svx_n, \svy_n] \in \mvx_n)\delta
    \]
    of the form $\neg\sigma(\mathsf{p}_{i})$ in
    (\ref{FORALLPIZERONU}), where $\nu(\mathsf{p}_{i}) = \false$, into
    the equisatisfiable formula
    \[
    \bigwedge_{i=1}^{h} \svx_i \in \svz_i 
    \wedge \bigwedge_{j=h+1}^{n} [\svx_j, \svy_j] \in \mvx_j
    \wedge \neg \delta \,.
    \]
    Let $\varphi'$ be the
    normalized \Lang-conjunction so obtained. Plainly, $\varphi' 
    \rightarrow \varphi$ is satisfied by any interpretation.
        
    
    \item Check that $\varphi'$ is satisfiable by a 
    \textsc{NExpTime} procedure for normalized \Lang-conjunctions 
    (cf.\ Theorem \ref{DEC}).
\end{itemize}
Since $\varphi'$ can be constructed in non-deterministic linear time, 
the corollary follows.
\end{proof}

Next we consider \LangBounded{h}-formulae, namely \Lang-formulae whose
simple-prenex subformulae have \corr{quantifier-prefix lengths} bounded by
the constant $h\geq 0$.
By reasoning much as in the proofs of Theorem~\ref{DEC} and Corollary
\ref{corollaryLangFormulae}, it is immediate to check that the s.p.\
for \LangBounded{h}-formulae can be reduced in non-deterministic
linear time to the s.p.\ of $\ForallpizeroBounded{h}$-formulae, and
thus, by \cite[Corollary~4]{CanLonNic2011}, it can be decided in
non-deterministic polynomial time.
On the other hand, it is an easy matter to show that the s.p.\ for
\LangBounded{h}-formulae is \textsc{NP}-hard.  Indeed, given a
propositional formula $Q$, consider the \LangBounded{0}-formula
$\psi_{Q}$, obtained from $Q$ by replacing each propositional variable
$\mathsf{p}$ in $Q$ with the atomic \Lang-formula $x_{\mathsf{p}} \in
X$, where $X$ and the $x_{\mathsf{p}}$'s are distinct set variables.
Plainly, $Q$ is propositionally satisfiable if and only if the
\Lang-formula $\psi_{Q}$ is satisfiable.
The following lemma summarizes the above considerations.

\begin{lemma}\label{NP}
For any integer constant $h \geq 0$, the s.p.\ for
\LangBounded{h}-formulae is \textsc{NP}-complete.  \qed
\end{lemma}

It is noticeable that, despite of the large collection of
set-theoretic constructs which are expressible by \Lang-formulae (see
Table \ref{SETCONS}), some very common map-related operators like
domain, range, and map image can not be expressed by \Lang-formulae in
full generality, but only in restricted contexts.  In the next section
we prove that dropping any of such restrictions triggers
undecidability.

\section{Some undecidable extensions of \Lang}\label{UNDEC}

In this section we prove the undecidability of any extension of \Lang
which allows one to express literals of the form $x \subseteq
\dom(f)$.  As we will see, analogous undecidability results hold also
for similar extensions of \Lang in the case of other map related
constructs such as range, map image, and map composition.
Our proof will be carried out via a reduction of the \emph{Domino
Problem}, a well-known undecidable problem studied in \cite{Ber1966}
\corr{(see also \cite{BorGraGur1997})},
which asks for a tiling of the quadrant $\nat \times \nat$ subject to
a finite set of constraints.

\begin{definition}[Domino problem]\label{DOMINO}
A \emph{domino system} is a triple $\dominoSys=(D,H,V)$, where
$D=\{d_1, \ldots, d_\ell\}$ is a finite nonempty set of \emph{domino
types}, and $H$ and $V$, respectively the \emph{horizontal} and
\emph{vertical compatibility conditions}, are two functions
which associate to each domino type $d \in D$ a subset of $D$,
respectively $H(d)$ and $V(d)$.

A \emph{tiling} $t$ for a domino system $\dominoSys=(D,H,V)$
is any mapping which associates a domino type in $D$ to each 
ordered pair of natural numbers in $\nat \times \nat$.
A tiling $t$ is said to be \emph{compatible} 
if and only if  $t [m+1,n] \in H(t[m,n])$ and 
$t [m, n+1] \in V(t[m,n])$ for all $n,m \in \nat$.
The \emph{domino problem} consists in 
determining whether a domino system admits a compatible tiling.
\qed
\end{definition}

In order to reformulate the domino problem in set-theoretic terms, we
make use of the following set-theoretic variant of Peano systems (see, 
for instance, \cite{Mos2005}).
%

\begin{definition}[Peano systems]\label{PEANO}
Let $\pi$ be a pairing-function and let $\peanN, \peanZ, \peanS$ be
three sets in the von Neumann hierarchy of sets.  The tuple
$\peanSys=(\peanN, \peanZ, \peanS, \pi)$ is said to be a Peano system
if it satisfies the following conditions:

\begin{enumerate}[label=\textbf{(P\arabic*)}]
 \item\label{P1} $\peanN$ is a set to which $\peanZ$ belongs;

 \item\label{P2} $\peanS \subseteq \peanN \times_{\pi} \peanN$ is a
 total function over $\peanN$, i.e., a single-valued map with domain
 $\peanN$;
 
 \item\label{P3} $\peanS$ is injective;
 
 \item\label{P4} $\peanZ$ is not in the range of $\peanS$;
 
 \item\label{P5} for each $X \subseteq \peanN$ the following holds:
\[
 (\peanZ \in X \wedge (\forall n \in \peanN)(n \in X \longrightarrow \peanS n \in X)) \longrightarrow X = \peanN.
\]
\end{enumerate}
\end{definition}
\noindent The first Peano system was devised by G.\ Peano himself.  It
can be \corr{characterized} as $\peanSys_{0}=(\peanN_{0}, \peanS_{0},
\emptyset, \pi_{\mathsf{Kur}})$, where $\peanN_{0}$ is the minimal set
containing the empty set $\emptyset$ and satisfying $(\forall \sx \in
\peanN_{0})(\{\sx\} \in \peanN_{0})$, and $\peanS_{0}$ is the relation
over $\peanN_{0}$ such that $\pi_{\mathsf{Kur}}(\sx, \sy) \in
\peanS_{0}$ if and only if $\sx \in \sy$.\footnote{In the original
definition the pairing function was not specified.}

The domino problem can be easily reformulated in pure set-theoretic
terms.  To this purpose, we observe that any tiling $t$ for a domino
system induces a partitioning of the integer plane $\nat \times \nat$,
as it associates exactly one domino type to each pair $\langle
n,m\rangle \in \nat \times \nat$.  Hence, given a domino system
$\dominoSys=(\{d_1, \ldots, d_\ell\},H,V)$, the domino problem for
$\dominoSys$ can be expressed in set-theoretic terms as the problem of
deciding whether there exists a partitioning $\Partition=(\partn_1,
\ldots, \partn_\ell)$ of $\peanN \times_{\pi} \peanN$, for some fixed
Peano system $\peanSys=(\peanN, \peanZ, \peanS, \pi)$, such that for
all $u,v,u',v' \in \peanN$, and for all $1 \leq i,j \leq \ell$ such
that $\pi(u,v) \in \partn_i$ and $\pi(u',v') \in \partn_j$,
\begin{enumerate}[label=\textbf{(D\arabic*)}]
\item\label{D1} if $\pi(u,u') \in \peanS$ (i.e., $u'$ is the successor
of $u$) and $v=v'$ then $d_j \in H(d_i)$, and

\item\label{D2} if $\pi(v,v') \in \peanS$ (i.e., $v'$ is the successor
of $v$) and $u=u'$ then $d_j \in V(d_i)$.
\end{enumerate}

Notice that from the properties of Peano systems it follows that if a
domino system $\dominoSys$ admits a compatible tiling $t$ then we can
construct a partitioning of the integer plane which satisfies \ref{D1}
and \ref{D2} however the Peano system is chosen.


All instances of the domino problem can be formalized with normalized
\Lang-conjunctions extended with two positive literals of the form $x
\subseteq \dom(\mvx)$, with $x \in \SetVars$ and $\mvx \in \MapVars$,
where the obvious semantics for the operator $\dom(\cdot)$ is 
$
\inter (\dom(\mvx)) \defAs  \{u \in \VNU : [u,v] \in \inter \mvx,
\text{ for some } v \in \VNU\}\,,
$
for any interpretation $\inter$.  In view of the undecidability of the
domino problem, this yields the undecidability of the s.p.\ for the
class $\Langdom$ of normalized \Lang-conjunctions extended with two
positive literals of the form $x \subseteq \dom(\mvx)$, proved in the
following theorem.

%
%
%
%
%

\begin{theorem}\label{PEANUNDEC}
The s.p.\ for $\Langdom$, namely the class of normalized
\Lang-conjunctions extended with two positive literals of the form $x
\subseteq \dom(\mvx)$, is undecidable.
\end{theorem}
\begin{proof}
Let $\dominoSys=(D,H,V)$, with $D=\{d_1, \ldots, d_\ell\}$, be a
domino system.  We will show how to construct in polynomial time a
formula $\varphi_{\dominoSys}$ of $\Langdom$ which is satisfiable if and only
if there exists a partitioning of the integer plane which satisfies
conditions \ref{D1} and \ref{D2}, so that the undecidability of the
s.p.\ for $\Langdom$ will follow directly from the undecidability
of the domino problem.

Let $\mN$, $\mZ$ be two distinct set variables, and let $\mS$ be a map
variable.  In addition, let $\mQ_{1}, \ldots, \mQ_{\ell}$ be pairwise
distinct map variables, which are also distinct from $\mS$.  These are
intended to represent the blocks of the partition of the integer plane
induced by a tiling.
To enhance the readability of the formula $\varphi_{\dominoSys}$ we are about to
construct, we introduce some abbreviations which will also make use 
of some map constructs defined in Table~\ref{SETCONS}. 
To begin with, we put
\[
  \isPartition(\mQ_1, \ldots, \mQ_\ell; \mN \times \mN) \defAs
  \mN \times \mN \subseteq \mQ_1 \cup \ldots \cup \mQ_\ell \,\,\, \wedge \,\,
  \bigwedge_{i\neq j} \left(\mQ_i \cap \mQ_j =
  \emptyset\right).
\]
Plainly, for every interpretation \inter, we have $\inter \models
\isPartition(\mQ_1, \ldots, \mQ_\ell; \mN \times \mN)$ if and only if
$(\inter \mQ_1, \ldots, \inter \mQ_\ell)$ partitions $\inter
(\mN \times \mN)$.
%
Next we define the formulae $\horcom_i$ and $\vercom_i$, for $i = 
1,\ldots,\ell$, which will encode respectively the horizontal and the
vertical compatibility constraints:
\[
  \horcom_i  \defAs  S^{-1} \circ \mQ_i \subseteq \bigcup_{d_{j} 
  \in H(d_{i})} \mQ_{j}\,, \qquad \qquad
  \vercom_i  \defAs  \mQ_i \circ S \subseteq \bigcup_{d_{j} 
  \in V(d_{i})} \mQ_{j}\,.
\]
%
%
Finally, we denote with $\isPeanoSys(N,Z,S)$ the following formula:
\[
  \isPeanoSys(\mN,\mZ,\mS) \defAs \mZ \in \mN \wedge \bij(\mS)
  \wedge \dom(\mS)=\mN \wedge \range(\mS) = (\mN \setminus \{\mZ\})
  \wedge (\forall [\svx, \svy] \in \mS)(\svx \in \svy).
\]
Notice that $\range(\mS) = (\mN \setminus \{\mZ\})$ is equivalent to
$\dom(\mS^{-1}) = (\mN \setminus \{\mZ\})$.  In addition, a literal of
the form $x = \dom(f)$ can obviously be expressed by the conjunction
$
(\forall [\svx', \svy'] \in \mvx)(\svx' \in \svx) \,\wedge \, x \subseteq
\dom(f)\,.
$

Next we show that the formula $\isPeanoSys(\mN,\mZ,\mS)$ is
satisfiable and correctly characterizes Peano systems, in the sense
that if $\inter \models \isPeanoSys(\mN,\mZ,\mS)$ for an interpretation
$\inter$, then $(\inter \mN, \inter \mZ, \inter \mS, \ipairf)$ is a
Peano system.
Given any interpretation $\inter$ such that $\inter \mN=\peanN_{0}$,
$\inter \mS=\peanS_{0}$, $\inter \mZ = \emptyset$, and $\ipairf =
\pi_{\mathsf{Kur}}$, $\inter \models \isPeanoSys(\mN,\mZ,\mS)$ follows
from the very definition of $\peanSys_{0}$, so that
$\isPeanoSys(\mN,\mZ,\mS)$ is satisfiable.
In addition, if $\inter \models \isPeanoSys(\mN,\mZ,\mS)$ for an
interpretation $\inter$, it can easily be proved that $(\inter \mN,
\inter \mZ, \inter \mS, \ipairf)$ is a Peano system.  Indeed \ref{P1},
\ref{P2}, \ref{P3}, and \ref{P4} follow readily from the first four
conjuncts of $\isPeanoSys(\mN,\mZ,\mS)$.  
Concerning \ref{P5}, we proceed by contradiction.
Thus, let us assume that there exists a proper subset $X$ of $\inter
\mN$ such that the following holds
\begin{equation}\label{UNDECEQ1}
 \inter \mZ \in X \wedge (\forall n,n' \in \inter \mN)\big((n \in X
 \wedge \ipairf(n,n') \in \inter \mS) \longrightarrow n' \in X\big)
\end{equation}
and let $\sx$ be a set in $\inter \mN \setminus X$ with minimal rank.
We must have $\sx\neq\inter \mZ$, in force of the first conjunct of
(\ref{UNDECEQ1}), and thus $\sx \in \range(\inter \mS)$ must hold, as
we assumed that $\inter$ correctly models the conjunct $\range(\mS) =
(\mN \setminus \{\mZ\})$ of the formula $\isPeanoSys(\mN,\mZ,\mS)$.
Hence, there must exist a set $\sy$ such that $\pi(\sy, \sx) \in
\inter \mS$.  Since $\inter \models (\forall [\svx, \svy] \in \mS)(\svx
\in \svy)$, $\sy$ must have rank strictly less than $\sx$, so that
$\sy \in X$ must hold, as by assumption $\sx$ has minimal rank in
$\inter \mN \setminus X$.  But (\ref{UNDECEQ1}) would yield $\sx \in X$,
which contradicts our initial assumption $\sx \in \inter X \setminus
N$.


We are now ready to define the formula $\varphi_{\dominoSys}$ of
$\Langdom$ intended to express that the domino system
$\dominoSys=(D,H,V)$ admits a compatible tiling. This is:
\[
 \varphi_{\dominoSys} \defAs \isPeanoSys(\mN,\mZ,\mS) \wedge
 \isPartition(\mQ_1, \ldots, \mQ_\ell; \mN \times \mN) \wedge
 \bigwedge_{i=1}^{\ell} \horcom_i
 \wedge
 \bigwedge_{i=1}^{\ell} \vercom_i\,.
\]
Observe that $\varphi_{\dominoSys}$ can be expanded so as to involve 
only two literals of the form $\svx \subseteq \dom(\mvx)$.

We show next that $\varphi_{\dominoSys}$ is satisfiable if and only if
the domino system $\dominoSys$ admits a compatible tiling.
Let us first assume that $\varphi_{\dominoSys}$ is satisfiable, and
let $\inter$ be a model for $\varphi_{\dominoSys}$.  Plainly, $(\inter
\mN, \inter \mZ, \inter \mS, \ipairf)$ is a Peano system, as
$\isPeanoSys(\mN,\mZ,\mS)$ is a conjunct of $\varphi_{\dominoSys}$.
In addition, $(\inter \mQ_1, \ldots, \inter \mQ_\ell)$ partitions
$\inter \mN \times \inter \mN$, since $\inter \models
\isPartition(\mQ_1, \ldots, \mQ_\ell; \mN \times \mN)$.  It remains to
prove that the partition $(\inter \mQ_1, \ldots, \inter \mQ_\ell)$ is
induced by a compatible tiling of the domino system $\dominoSys$,
i.e., that properties \ref{D1} and \ref{D2} hold.  Thus let $u,u',v
\in \inter \mN$ such that $\ipairf(u,v) \in \inter \mQ_i$,
$\ipairf(u',v) \in \inter \mQ_j$, and $\ipairf(u,u') \in \inter S$,
for some $1 \leq i,j \leq \ell$.  Plainly $\ipairf(u',v) \in \inter
(S^{-1} \circ \mQ_i)$, so that from $\inter \models \horcom_i$ it
follows $d_j \in H(d_i)$, proving \ref{D1}.  Likewise, let $u,v,v' \in
\inter \mN$ be such that $\ipairf(u,v) \in \inter \mQ_i$,
$\ipairf(u,v') \in \inter \mQ_j$, and $\ipairf(v,v') \in \inter S$,
for some $1 \leq i,j \leq \ell$.  Thus $\ipairf(u,v') \in \inter
(\mQ_i \circ \mS)$, so that from $\inter \models \vercom_i$ we obtain
$d_j \in V(d_i)$, proving \ref{D2}.

Conversely, let us suppose that $\dominoSys$ admits a compatible
tiling and let $(\partn_1, \ldots, \partn_\ell)$ be the induced
partitioning of $\peanN_{0} \times \peanN_{0}$ which
satisfies \ref{D1} and \ref{D2}, relative to the Peano system
$\peanSys_{0}=(\peanN_{0}, \emptyset, \peanS_{0},
\pi_{\mathsf{Kur}})$.
We prove that $\varphi_{\dominoSys}$ is satisfied
be any interpretation $\inter$ such that
\[
  \ipairf = \pi_{\mathsf{Kur}}\,, \qquad 
  \inter \mN = \peanN_{0}\,, \qquad 
  \inter \mZ = \emptyset\,, \qquad 
  \inter \mS = \peanS_{0}\,, \qquad 
  \inter \mQ_i = \partn_i\, ~(\text{for } i=1,\ldots,\ell)\,.   
\]
Plainly, $\inter$ models correctly $\isPeanoSys(\mN,\mZ,\mS)$.  In
addition, $\inter \models \isPartition(\mQ_1, \ldots, \mQ_\ell, \mN
\times \mN)$, as we assumed that $(\inter \mQ_1, \ldots, \inter
\mQ_\ell)=(\partn_1, \ldots, \partn_\ell)$ is a partitioning of
$\inter \mN \times \inter \mN = \peanN_{0} \times \peanN_{0}$.
Next we prove that $\inter$ models correctly the conjuncts $\horcom_i$
of $\varphi_{\dominoSys}$, for $i = 1,\ldots, \ell$.  To this purpose,
let $u$, $v$ be any two sets such that $\ipairf(u,v) \in \inter
(\mS^{-1} \circ \mQ_i)$, for some $1 \leq i \leq \ell$.  Then, there
must exist a set $u'$ such that $\ipairf(u',v) \in \inter \mQ_i$, and
$\ipairf(u',u) \in \inter \mS = \peanS_{0}$.  Hence $\ipairf (u,v)$
must belong to some $\partn_j = \inter \mQ_{j}$, for $1 \leq j \leq
\ell$, such that $d_j \in H(d_i)$, proving $\inter \models \horcom_i$.
Analogously, one can show that $\inter \models \vercom_i$, for $i =
1,\ldots, \ell$, thus proving that $\inter \models
\varphi_{\dominoSys}$ and in turn concluding the proof of the theorem.
\end{proof}

Because of the large number of set-theoretic constructs expressible in
\Lang, the undecidability of various other extensions of normalized
\Lang-conjunctions easily follows from Theorem \ref{PEANUNDEC}.

\begin{corollary}\label{OTHERUNDEC}
The class of normalized \Lang-conjunctions extended with two
literals of any of the following types is undecidable:
\begin{equation}
    \label{typesLiterals}
 \svx \subseteq \range(\mvx)\,, \qquad 
 \mvz \subseteq \mvx \circ \mvy\,, \qquad
 \svy \subseteq \mvx[\svx]\,,
\end{equation}
where $\svx, \svy \in \SetVars$ and $\mvx, \mvy, \mvz \in \MapVars$.
\end{corollary}
\begin{proof}
In view of Theorem \ref{PEANUNDEC}, it is enough to show that any
literal of the form $\svx \subseteq \dom(\mvx)$ can be expressed with
normalized \Lang-conjunctions extended with \emph{one} literal of any
of the types (\ref{typesLiterals}).  Concerning the case of literals
of the types
$\svx \subseteq \range(\mvx)$, $\mvz \subseteq \mvx \circ \mvy$
it suffices to observe that $\svx \subseteq \dom(\mvx)$ is equivalent
to each of the two formulae
$\svx \subseteq \range (\mvx^{-1})$ and 
$\identity{\svx} \subseteq \mvx \circ \mvx^{-1}$,
and that map identity $\identity{\svx}$ and map inverse $\mvx^{-1}$ 
are expressible by \Lang-formulae, as shown in Table~\ref{SETCONS}.

Finally, concerning literals of the form $\svy \subseteq \mvx[\svx]$,
it is enough to observe that for every set variable $R_{\mvx}$
distinct from $\svx$ we have
\begin{itemize}
    \item $\inter \models \svx \subseteq \mvx^{-1}[R_{\mvx}] \rightarrow 
    x \subseteq \dom(\mvx)$, for every interpretation $\inter$;
    
    \item if $\inter \models x \subseteq \dom(\mvx)$, for some
    interpretation $\inter$, then $\interp \models \svx \subseteq
    \mvx^{-1}[R_{\mvx}]$, where $\interp$ is the
    $\{R_{\mvx}\}$-variant of $\inter$ such that $\interp R_{\mvx} =
    \range(\inter \mvx)$.
\end{itemize}
Therefore, a $\Langdom$-formula $\psi \defAs \varphi \wedge \svx
\subseteq \dom(\mvx) \wedge \svy \subseteq \dom(\mvy)$, where $\varphi$
is a normalized \Lang-conjunction, is equisatisfiable with $\varphi
\wedge \svx \subseteq \mvx^{-1}[R_{\mvx}] \wedge \svy \subseteq 
\mvy^{-1}[R_{\mvy}]$, where $R_{\mvx}$ and $R_{\mvy}$ are two fresh 
distinct set variables not occurring in $\psi$.
%
\end{proof}

In the proof of Theorem \ref{DEC} we provided a reduction of the s.p.\
for normalized \Lang-conjunctions to the s.p.\ for
\Forallpizero-formulae.  Therefore, the undecidability results of 
Theorem~\ref{PEANUNDEC} and Corollary~\ref{OTHERUNDEC} hold also for 
the corresponding extensions of \Forallpizero-formulae.

\section{Conclusions and plans for future works}
\label{CONC}

In this paper we presented a quantified sublanguage of set theory,
called \Lang, which extends the language $\forall_{0}$ studied in
\cite{BreFerOmoSch1981} with quantifiers involving ordered pairs.  We
reduced its satisfiability problem to the same problem for formulae of
the fragment studied in \cite{CanLonNic2011}.  The resulting decision
procedure runs in non-deterministic exponential time.  However, if one
restricts to formulae with quantifier prefixes of length bounded by a
constant, the decision procedure runs in non-deterministic polynomial
time.  It turns out that such restricted formulae still allow one to
express a large number of useful set-theoretic constructs, as reported
in Table \ref{SETCONS}.  
Finally, we also proved that by slightly extending \Lang-formulae 
with few literals (at least two) of any of the types $\svx
\subseteq \dom(\mvx)$, $\svx \subseteq \range(\mvx)$, $\svx \subseteq
\mvx[\svy]$, and $h \subseteq f \circ g$, one runs into 
undecidability.

Other extensions of \Lang are to be investigated, in particular those
involving the transitive closure of maps.  Also, the effects of
allowing nesting of quantifiers should be further studied, extending 
the recent results \cite{OmoPol2010, OmoPol2012} to our context.

In contrast with description logics, the semantics of our language is
\emph{multi-level}, as most of the languages studied in the context of
Computable Set Theory. This characteristic may play
a central role when applying set-theoretic languages to knowledge
representation, with particular reference to the \emph{metamodeling} issue
(see \cite{WelFer1994, Mot2007}), which affects the description logics
framework.
However, the multi-level feature is limited in \Lang-formulae, since
clauses like $\mvx \in \svx$, $[\mvx, \mvy] \in \mvz$, 
with $\svx$ a set variable and $\mvx$, $\mvy$, and $\mvz$ map
variables, are not expressible in it.
In light of this, we intend to investigate extensions of the theory
\Lang which also admit constructs of these forms, and study
applications of these in the field of knowledge representation.

Finally, we intend to study correlations between our language \Lang
and \emph{Disjunctive Datalog} (cf.\ \cite{EitGotMan1997}) in order to
use some of the machinery already available for the latter to
simplify the implementation of an optimized satisfiability test for
the whole fragment \Lang, or just for a \emph{Horn-like} restriction
of it.

\section*{Acknowledgments}
The authors would like to thank the reviewers for their valuable 
comments and suggestions.

\bibliographystyle{eptcs}

\end{document}